\documentclass[acmsmall]{acmart}

\usepackage{hyperref}
\usepackage{amsmath}
\usepackage{algorithm}
\usepackage{algpseudocode}
\newtheorem{theorem}{Theorem}
\newtheorem{definition}{Definition}

\usepackage{xcolor}
\usepackage{overpic}

\usepackage{float}
\usepackage{pgfplots}
\pgfplotsset{width=7cm}
\usepackage{pgfplotstable}
\usepackage{sfmath}

\usepackage{listings}
\usepackage{graphicx}
\graphicspath{ {assets/} }

\usepackage{comment}

\AtBeginDocument{%
  \providecommand\BibTeX{{%
    \normalfont B\kern-0.5em{\scshape i\kern-0.25em b}\kern-0.8em\TeX}}}

\setcopyright{none} 
\copyrightyear{2024}
\acmYear{2024}

\hypersetup{pdftitle={Inexactness and Correction of Floating-Point Reciprocal, Division and Square Root},pdfauthor={Dutton, Anand, Enenkel, M\"{u}ller}}

\begin{document}
\title{Inexactness and Correction of \\
 Floating-Point Reciprocal, Division and Square Root}


\author{Lucas M. Dutton}
\email{dutton@mcmaster.ca}
\author{Christopher Kumar Anand}
\email{anandc@mcmaster.ca}
\affiliation{%
  \institution{McMaster University}
  \city{Hamilton}
  \state{Ontario}
  \country{Canada}
}
\author{Robert Enenkel}
\email{enenkel@ca.ibm.com}
\affiliation{%
  \institution{IBM Canada}
  \city{Markham}
  \state{Ontario}
  \country{Canada}
}
\author{Silvia Melitta M\"{u}ller}
\email{smm@de.ibm.com}
\affiliation{%
  \institution{IBM Germany Development}
  \city{Boeblingen}
  \state{Baden-W\"{u}rttemberg}
  \country{Germany}
}

\renewcommand{\shortauthors}{Dutton, Anand, Enenkel, M\"{u}ller}
\renewcommand{\shorttitle}{Correction of Division and Square Root}

\begin{abstract}
Floating-point arithmetic performance determines the overall performance of important applications, from graphics to AI. Meeting the IEEE-754 specification for floating-point requires that final results of addition, subtraction, multiplication, division, and square root are correctly rounded based on the user-selected rounding mode. A frustrating fact for implementers is that naive rounding methods will not produce correctly rounded results even when intermediate results with greater accuracy and precision are available. In contrast, our novel algorithm can correct approximations of reciprocal, division and square root, even ones with slightly \emph{lower} than target precision.

In this paper, we present a family of algorithms that can both increase the accuracy (and potentially the precision) of an estimate and correctly round it according to all binary IEEE-754 rounding modes. We explain how it may be efficiently implemented in hardware, and for completeness, we present proofs that it is not necessary to include equality tests associated with round-to-nearest-even mode for reciprocal, division and square root functions, because it is impossible for input(s) in a given precision to have exact answers exactly midway between representable floating-point numbers in that precision. In fact, our simpler proofs are sometimes stronger.
\end{abstract}

\maketitle

\ccsdesc[500]{Mathematics of computing~Arbitrary-precision arithmetic}
\ccsdesc[500]{Hardware~Arithmetic and datapath circuits}
\ccsdesc[300]{Theory of computation~Theory and algorithms for application domains}

\keywords{special function, instruction set architecture, floating-point}

\section{Introduction}
Elementary mathematical functions are widely used in high-performance
applications, such as AI, databases and scientific computing.
The precision and accuracy requirements for these functions differ for each
application.
Hardware implementations should ideally support high levels of both performance and accuracy,
while staying within area budgets. 
The highest level of accuracy is correctly rounded to IEEE-754 specifications \cite{IEEEStandardFloatingPoint2019}.
Ideally, we could meet this requirement while simultaneously achieving very high throughput and low latency.
In this paper,
we develop a method that can help with this goal by replacing the final rounding of algebraic functions with a correction algorithm.
The ability to correct small errors eliminates the need to calculate values to a higher precision and subsequently round them,
something more expensive than one might imagine due to IEEE-754's round-to-nearest-even mode.

\medskip

Most algorithms for evaluating algebraic functions (i.e., roots of polynomial equations, which include square root and divide) can be separated into
three phases: computing an initial approximation from the input,
refining the approximation to a sufficient accuracy, then
rounding the refined output to the target precision.
The initial approximation lacks the required accuracy and may have lower precision. 
Refining the approximation increases
the accuracy of the result, usually to the point
where we have excess accuracy (and therefore must use excess precision, meaning more bits than in the output format).

The rounding step must round the result to the desired precision
while meeting IEEE-754 requirements for rounding of inexact results.
In particular, the most common rounding mode requires that 
exact results which are midway between representable floating-point numbers
be rounded to the one whose final bit is zero.
In this work, in addition to the correction algorithm, 
we provide results that are useful for accelerating
the rounding step for reciprocal, divide and square-root (sqrt) functions,
by allowing the algorithm designer to
safely disregard cases that we show cannot occur.

\medskip
Section \ref{sect:background} sets the stage with required definitions and notation.
 Section \ref{sect:theorems} introduces and explains
  the new theorems, with proofs, that the set of real numbers obtained by taking the reciprocal, square root, reciprocal square root or division of numbers representable in the target floating-point format does not contain numbers at the midpoint between successive floating-point numbers.
In Section \ref{sect:implementation}, we present a novel correction procedure for these functions, which is simplified as a result of the aforementioned theorems.  We also describe how to efficiently implement this algorithm using common hardware blocks already present in floating-point units.  
Finally,  in Section \ref{sect:priorwork}, 
we contrast our implementation and theory with previously known results.

\medskip
\section{Background\label{sect:background}}
We first introduce terminology used to describe floating-point numbers,
and the choices made in this paper.
In particular, 
we define units in the last place (ulps),
and rounding modes.
Readers new to floating-point arithmetic should consult  \cite{goldbergWhatEveryComputer1991}.

\subsection{Floating-Point Numbers}
A floating-point number is like a number in scientific notation,
and is defined by the following parameters:
\begin{enumerate}
\item the base of the representation, which is always $2$ in this paper, since we only consider binary representations.
\item $n$, the precision, i.e. the number of bits in the significand in the binary representation, equivalently, $n-1$ is the number of fraction bits.
\item $e_{min}$ and $e_{max}$, the minimum  and maximum allowable exponents.
\end{enumerate}

\noindent
For a floating-point number $x$, we can write it as
\begin{equation} \label{eq:fpformat}
x = \pm d_0.d_1d_2... d_{n-1} \times 2^e
\end{equation}
such that $e_{\min} \leq e \leq e_{max}$ and $d_{i} \in \{0, 1\}$.
In optimized binary representations, such as IEEE-754,  we store
\begin{enumerate}
\item A sign bit $s$, indicating the sign, positive (0) or negative (1).
\item (Unsigned) exponent bits representing the non-negative number $e +e_{\text{bias}}$,
where the bias value is specified by the format, and given in Table~\ref{tab:formats} below.  Note the 0 value is reserved for subnormal numbers and exceptions, which can
  be handled using known techniques, and so will not be discussed in this paper.
\item significand bits, sometimes called mantissa bits, $m_{\text{bits}}$ represented similarly
  to Equation~\ref{eq:fpformat},
  but there is an implicit leading one bit that is not stored.
  The stored bits are called the fraction bits.
  This implicit one is equivalent to always setting bit $d_0=1$ in
  Equation~\ref{eq:fpformat}. Thus, the precision of the floating
  point number is always $n$.
\end{enumerate}

\begin{figure}[ht]
  \centering
  \includegraphics[width=1.0\textwidth]{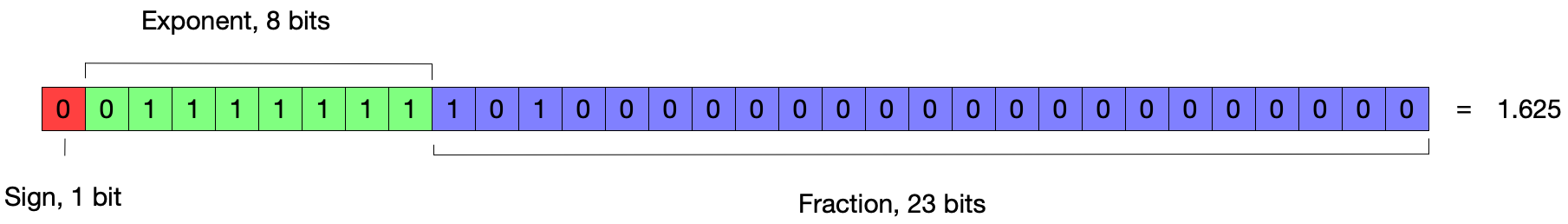}
  \caption{single-precision binary floating-point layout}
  \label{fig:spexample}
\end{figure}

There are several proposed formats of floating-point numbers. The most widely
used formats are contained in the IEEE-754 Standard for Floating-Point Arithmetic
\cite{IEEEStandardFloatingPoint2019}, especially 32-bit single-precision and
64-bit double-precision numbers. 
For normal numbers (i.e., when the exponential bits are not zero), we can calculate the represented value as
\begin{equation} \label{eq:fpbinary}
x = -1^{s_{\text{bit}}} \times 2^{e_{\text{bits}} - e_{\text{bias}}} \times \left( 1 + \sum_{i=1}^{n} m_{\text{bits}}[i] * 2^{-i}\right).
\end{equation}
Figure~\ref{fig:spexample}
illustrates an example of a single-precision number, whose value via
 \eqref{eq:fpbinary} is
\[
-1^{0} \times 2^{127 - 127} \times \left(1 + (2^{-1} + 2^{-3})\right) = 2^0 \times \left(1 + 0.5 + 0.125\right) = 1.625.
\]

Recently, AI researchers have introduced new and alternative 
lower-precision formats, allowing for faster but less accurate computations.
Table~\ref{tab:formats} summarizes some common formats.

\begin{table}[h]
  \hfil
  \caption{\label{tab:formats}
    Standard and some non-standard binary floating-point formats in use today.
  }
  \begin{tabular}{|p{3cm}|r|r|r|r|}
    \hline
    Format & Exponent & Fraction & Total bits & Bias \\
    \hline
    IEEE Double & 11 & 52 & 64 & 1023 \\
    \hline
    IEEE Single & 8 & 23 & 32 & 127 \\
    \hline
    IEEE Half & 5 & 10 & 16 & 7 \\
    \hline
    DLFloat \cite{agrawalDLFloat16bFloating2019a}
      & 6 & 9 & 16 & 32 \\
    \hline
    bfloat16 \cite{kalamkarStudyBFLOAT16Deep2019}
      & 8 & 7 & 16 & 127 \\
    \hline
  \end{tabular}\hfil
  \vspace{2mm}
\end{table}

\subsection{Unit in Last Place}
Intuitively, we can think of an ulp to be the place value of the lowest-order bit
of the significand. 
There are several formal definitions in the literature
such as \citet{harrisonMachineCheckedTheoryFloating1999b},
\citet{muller2005definition} and \citet{kahanLogarithmTooClever2004}.
We follow closely the definition given in
\citet{cornea-haseganCorrectnessProofsOutline1999a}.

\begin{definition}[Unit in Last Place]
  Let $x = d_0.d_1d_2... d_{n-1} \times 2^e$ be a floating-point number. The
  \texttt{ulp} of $x$ is

\[
ulp(x) = 0.\underbrace{0 \cdots 00}_\text{\makebox[0pt]{$n-2$ leading zeros}}1 \times 2^e
\]

\end{definition}

Note that in this definition, floating-point numbers have an ulp, but non-representable numbers do not.
This distinction is important for numbers just less than a power of two, 
for which the ulp values of the bracketing floating-point numbers change.
All of the numbers with the same exponent bits, and therefore the same ulp value are called a \emph{binade}.
In practice, hardware implementations calculate exponent and significand values independently
and calculate the significand values using fixed-point arithmetic.
Thus, the ulp value is the place value of a particular bit.

It is convenient to define errors in terms of ulps of the correctly rounded result. 
For example, if we allow for 
2 ulps of error, we have 5 possible values that the result can be rounded to,
in the set $\{-2, -1, 0, 1, 2\}$.
To be precise, we need to modify this definition for values near the binade boundary,
but usually, these values are automatically excluded by the properties of the algorithms.

\subsection{Rounding}

\begin{figure}[h!]
  \centering
  \includegraphics[width=0.5\textwidth]{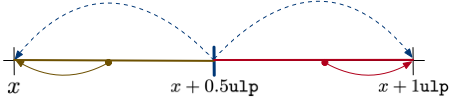}
  \caption{Cases for rounding to nearest even}
  \label{fig:ulpline}
\end{figure}

Rounding is important because no set of representable floating-point numbers in a given precision is closed under division, 
let alone general algebraic functions.  
For example, the number $1/10$ can be represented
exactly in decimal, but not in a binary format:
it will have an infinite sequence of bits. 
We must round this number to
a representable floating-point number, which depends on the rounding mode.

The default rounding mode used is round-to-nearest-even, and there are 3 cases
to deal with when rounding with this mode, as shown in Figure~\ref{fig:ulpline}:
\begin{itemize}
\item The number is exactly representable: No rounding is necessary.
\item The number is in the upper(red)/lower(yellow) half interval: Rounds to the nearest representable number.
\item At the midpoint: Round to the \textbf{nearest even} representable number,
  illustrated with the blue dotted lines.
\end{itemize}

The midpoint case makes round-to-nearest-even the hardest to treat in most implementations. 
The other rounding modes, which are round-to-$+\infty$, round-to-$-\infty$ and round toward zero, are less complicated:

\begin{itemize}
\item Round-to-$+\infty$: rounds up toward positive infinity, works as the \textit{ceiling} function
\item Round-to-$-\infty$: rounds down toward negative infinity, works as the \textit{floor} function
\item Round-to-zero: rounds toward zero, works as the \textit{truncate} function
\end{itemize}

\medskip
\section{Rounding Theorems\label{sect:theorems}}
The rounding modes explained in Section~\ref{sect:background} provide a
specification that designers must satisfy when implementing an algorithm involving
floating-point numbers. 
In particular, round-to-nearest-even requires an equality check which the other modes do not require. 
For elementary function implementations, which
typically rounds a refined, high-precision approximation down to its resulting precision,
this requirement risks increasing the execution time of the rounding procedure.

Our proofs in this section will demonstrate that
for the reciprocal, division and square root functions, 
the exact result will never land at a midpoint when the operand
and the result are in the same floating-point format. This
allows implementations to remove midpoint checks from the underlying
hardware, and simplifies the implementation.

The results we need are not new.  
Some go back to \citet{markstein2000ia},
and \citet{jeannerodMidpointsExactPoints2011}
have cataloged cases in which many common functions can produce results that are representable or hit the midpoint between representable numbers in a target precision,
not only for binary floating-point but for other small radices. 

We present our proofs for the results we need to make this paper complete,
to present simpler binary-only proofs,
and to point out that a stronger result is true: 
not only are midpoints never the exact reciprocal of a floating-point number---reciprocals of floating-point numbers are either representable in the same precision or they are not representable in any precision, i.e., their binary expansions are always infinite.

\medskip
We also point out that while we do need the theorems about midpoints,
we do not need theorems about exactly representable reciprocals, divisions or square roots.
Our implementation in Section~\ref{sect:implementation} does not
  require a test for equality because:
\begin{enumerate}
\item in round-to-even mode, we know that midpoints cannot occur as exact answers; and
\item for all modes, since we start with a result in the target precision, the correction is always an integral number of ulps, and its calculation will be the same, whether the exact result is representable or not.
\end{enumerate}

\subsection{Representable Numbers}
Informally, we defined representable numbers as numbers that have an exact
floating-point representation. 
Formally,
\begin{definition}
For a given precision $n$, a
real number $x\in(1,2)$ is \emph{representable} if and only if:

\begin{align}
\exists B,e \in \mathbb{Z}\quad:\quad & 2^{n-1} \le B < 2^{n}, \\
                            & e_{min}\le e \le e_{max},  \\
          \text{such that } \quad\quad                  & x = \frac{B}{2^{n-1}}\times 2^{e}.
\end{align}
\end{definition}

If we restrict our attention to numbers in the interval $(1,2)$, $e=0$ can be omitted.
In other words, representable numbers can be expressed as a fraction
with an $n$-bit binary positive integer whose leading bit is $1$ and which contains at least one additional $1$ in the numerator, and a power of two in the denominator.
In particular, $B$ cannot be a power of $2$, but it can be divisible by powers of $2$.
The largest such power occurs for $110...0=2^{n-1}+2^{n-2}$, which is divisible by $2^{n-2}$.

\subsection{Reciprocals cannot be midpoints}
Floating-point numbers are either powers of two, in which case the reciprocal is another power of two,
or their binary representations require an infinite series of fractional bits,
i.e., they are not representable in any finite precision, 
and therefore they cannot be branch points for any rounding mode,
including the problematic midpoints between representable floating-point numbers. 
 
\begin{theorem}
\label{thm:reciprocal}
Let $n>1$ be a floating-point precision and $A \in (2^{n-1}, 2^n)$ the numerator of the representable number $\frac{A}{2^{n-1}}$. 
There does not exist $m>1$ and $B \in (2^{m-1}, 2^m)$ such that $\frac{B}{2^{m}} = \frac{2^{n-1}}{A}$, i.e, the inverse of $\frac{A}{2^{n-1}}$.
\end{theorem}

\begin{proof}
  We use proof by contradiction. Assume that such $m$ and $B$ exist.  It follows that
\begin{align*}
   \frac{B}{2^{m}} = \frac{2^{n-1}}{A} 
  &\iff AB = 2^{n+m-1}  \\
\end{align*}
but, by the condition of the theorem, $A$ is not a power of $2$ 
and by the fact that the image under inverse of the interval $(1,2)$ is the interval $(0.5,1)$, $B$ is also not
a power of $2$ which is only possible if $n+m-1=0$, which is a contradiction.
\end{proof}

From the above theorem, we can deduce that the only floating-point numbers with representable inverses are powers of two.
Consider the input binade $[1,2)$ with exponent $0$.
The included point, $1$, is a fixed point of reciprocal,
and the interior points are mapped to the interval $(1/2,1)$.
It follows that the exponent is mapped as 
\begin{equation}
    e \mapsto \begin{cases}
        -e & \text{input is a power of two, or} \\
        -e - 1 & \text{otherwise.}
    \end{cases}
\end{equation}
It follows that calculating inverses can be performed in
fixed-point arithmetic for the interior of the binades,
which are exactly all non-power-of-two inputs and outputs.

This leads us to the corollary, which concerns
round-to-nearest-even cases. 
It states that the reciprocal of
a floating-point number cannot be at the midpoint between
two floating-point numbers.
\begin{corollary}[Reciprocals cannot be midpoints]
\label{cor:recipMidpoint}
Let $n>1$ be a floating-point precision and $A \in (2^{n-1}, 2^n)$. Then $\frac{2^{n-1}}{A} + 2^{-n-1}$ is not representable.
\end{corollary}

\begin{proof}
A midpoint between precision-$n$ floating-point numbers is a precision-$n+1$ floating-point number.
So this would contradict the theorem.
\end{proof}
This makes the equality test in round-to-nearest-even implementations unnecessary.

\subsection{Divisions cannot be midpoints}
A similar result follows for the midpoint for division.
The quotient of two floating point numbers cannot be a midpoint in any equal or higher precision.
So, in particular, dividing two single-precision numbers cannot result in a midpoint in single or double-precision.

\begin{theorem}
  Let $n > 1$ be a floating-point precision, 
  $$
  \frac{A}{2^{n-1}}\times 2^p, \quad 
  \frac{B}{2^{n-1}}\times 2^q
  $$
  be floating point numbers, where integers $A,B \in [2^{n-1},
  2^n)$. 
  There does not exist $m \ge n$ and floating-point number
  $$
  \frac{C}{2^{m-1}}\times 2^r,
  $$
  where integer
  $C \in
  (2^{m-1}, 2^{m})$,
  such that 
  $$\frac{A/2^{n-1}\times 2^p}{B/2^{n-1} \times 2^q} = 
    \left( \frac{C}{2^{m-1}} + \frac{1}{2^{m}} \right) \times 2^r.$$
\end{theorem}

\noindent
Notes:
\begin{enumerate} 
  \item When $n=m$, this result was obtained by \citet{markstein2000ia}.
  \item This theorem applies to the division of numbers with different precisions by taking $n$ as the maximum precision of the two.
  \item If $C/2^{m-1}\times 2^r$ is the largest representable floating-point number smaller than the quotient,
  then the exponent can be computed by
  $$r = \begin{cases} 
          p - q & \text{if } A \geq B,\\
          p - q - 1 & \text{if } B > A.
        \end{cases}
  $$
  \item The case $A=B$ results in $1$, so there is nothing to prove.
  \item The case $A=2^{n-1}$ reduces to reciprocal for which the result follows from the previous section.
  \item The case $B=2^{n-1}$ only concerns the exponent.  The significand is copied from the input to the output.
\end{enumerate}

Unlike the case for reciprocal, we must consider two cases in the proof:  the numerator's significand being larger than the denominator's and vice versa.
As in the proofs in the reciprocal case, we can assume without loss of generality that the inputs are in the interval $[1,2)$, so we have integers $A, B \in
[2^{n-1}, 2^n)$, and floating-point numbers $a = \frac{A}{2^{n-1}}$ and $b = \frac{B}{2^{n-1}}$.
The included endpoints correspond to an input being $1$.
If the denominator is $1$, the quotient is simply the numerator.
If the numerator is $1$, 
the quotient is a reciprocal,
so the result follows from Theorem~\ref{thm:reciprocal}.
We split the proof of the remaining cases into two lemmas,
one for  $A<B$ and one for $A>B$, 
and the theorem follows from the two lemmas.

\begin{lemma}
  \label{thm:midpoint-lt}
  Let $n > 1$ be a floating-point precision, and $A,B \in (2^{n-1},
  2^n)$ and $A<B$. There does not exist $m \ge n-1 $ and $C \in
  (2^{m-1}, 2^m)$  such that $\frac{A}{B} = \frac12 \left(\frac{C}{2^{m-1}} + \frac{1}{2^{m}}\right)$,
  where the factor of $1/2$ is to adjust for the fact that $A/B\in[1/2,1)$.
\end{lemma}

\begin{lemma}
  \label{thm:midpoint-gt}
  Let $n > 1$ be a floating-point precision, and $A,B \in (2^{n-1},
  2^n)$ and $A>B$. There does not exist $m \geq n - 1$ and $C \in
  (2^{m-1}, 2^m)$  such that $\frac{A}{B} = \frac{C}{2^{m-1}} + \frac{1}{2^m}$.
\end{lemma}

These two lemmas imply that the result of dividing two floating-point numbers (as real numbers)
can never produce a number at the exact midpoint between two representable floating-point numbers.
This means that when rounding a result of the division of two floating-point numbers,
we never have to consider the midpoint in round-to-nearest-even cases.

The bounds on these lemmas are sharp, as demonstrated by the following example
with $n=4$ and $m=2$,
 where we take $A = 1.875$ and $B = 1.5$,
the quotient is a midpoint for $2 = m \not\geq n - 1 =4-1 = 3$.

\begin{align} 
& 1.5 \times 1.25 =  1.875 \\
\intertext{and expressing $A$, $B$ and $C$ in binary is}
\iff & \frac{1100_2}{2^3} \times \frac{1010_2}{2^3} = \frac{1111_2}{2^3} \\
\iff & \frac{1111_2}{2^3} / \frac{1100_2}{2^3} = \frac{10_2}{2^1} + \frac{1}{2^2},
\end{align}
which would be a counterexample if the theorem did not contain the $m\ge n-1$ clause.

\begin{proof}[Proof of Lemma~\ref{thm:midpoint-lt}]
\textit{Case $m\ge n$:} We prove this by contradiction. 
Assume that the following
equation holds:

\begin{equation}\label{eq:AoverB}
\frac{A}{B} = \frac{C}{2^m} + \frac{1}{2^{m+1}}
\end{equation}

This can be further simplified:
\begin{align}
  \iff & 2^{m+1}A = 2BC + B \\
  \iff & 2^{m+1}A = B(2C + 1) \\
\intertext{but since $2C+1$ must be odd, it cannot contain powers of $2$, so $B$ must be divisible by $2^{m+1}$,}
  \implies & 2^{m+1} | B. \label{eq:twoDividesB}
  \intertext{However, $B$ is in the interval $(2^{n-1}, 2^n)$, so the largest power of $2$ it can contain occurs when $B=2^{n-1}+2^{n-2}=3\cdot2^{n-2}$.
  Substituting this into equation~\eqref{eq:twoDividesB} yields $2^{m+1} | 3 \cdot 2^{n-2}$, 
  }
   \implies & m+1 \le n-2 \\
   \iff & m \le n-1.
\end{align}

\textit{This leaves the case $m= n-1$:} We prove this by contradiction. 
For this case, we decompose
\begin{align*}
    B &= 2^{n-1} + B',
\end{align*}
where $B'\in (0,2^{n-1}) = (0,2^m)$.
With this decomposition, equation~\eqref{eq:AoverB} becomes
\begin{equation}\label{eq:AoverBprime}
\frac{A}{2^m+B'} = \frac{C}{2^m} + \frac{1}{2^{m+1}}
\end{equation}

This can be further simplified:
\begin{align*}
  \iff & 2^{m+1}A = 2C(2^m+B') + (2^m + B') \\
  \iff & 2^{m+1}A = 2^{m+1}C + 2B'C + 2^m + B' \\
  \iff & 2^{m+1}A -2^{m+1}C-2^m= B'(2C + 1) \\
  \iff & 2^{m}(2A -2C-1)= B'(2C + 1) \\
\intertext{but since $2C+1$ must be odd, it cannot contain powers of $2$, so $B'$ must be divisible by $2^{m}$,}
  \implies & 2^{m} | B'.
\end{align*}
But this contradicts the fact that $B'\in(0,2^m)$.
\end{proof}

\begin{proof}[Proof of Lemma~\ref{thm:midpoint-gt}]
Again, we prove by contradiction. Then we assume that the following
equation holds:

\[
\frac{A}{B} = \frac{C}{2^{m-1}} + \frac{1}{2^{m}}
\]

This can be further simplified:
\begin{align*}
     & \frac{A}{B} = \frac{C}{2^{m-1}} + \frac{1}{2^{m}} \\
     \iff & 2^mA = 2BC + B \\
     \iff & 2^mA = B(2C + 1) \\
     \intertext{but since $2C+1$ must be odd, it cannot contain powers of two, so $B$ must be divisible by $2^m$,}
     \implies & 2^m | B.
     \intertext{However, $B$ is in the interval $(2^{n-1}, 2^n)$, so the largest power of $2$ it can contain occurs when $B=2^{n-1}+2^{n-2}$,}
   \implies & m \le n-2. \\
\end{align*}


\end{proof}

\subsection{Theorems on Square Roots}
We already have the following property on square root:
If the input is not a perfect square, then the result will
be irrational, which implies irrepresentability in any precision.
Thus, we only need to prove that perfect squares
do not occur on midpoints.

We define perfect squares in a set of numbers as a number that can be written
as a square of another number in the same set.
For example, in the rationals,

\[
\left\{ y \in \mathbb{Q} : \exists x \in\mathbb{Q}, y = x^2\right\}.
\]

We can specialize the rational numbers to our floating-point representation,
and similarly to the previous cases take advantage of the fact that the square root preserves binades---this time mapping two adjacent binades to one.
We only need to prove our result on binades $[1,2)$ and $[2,4)$ which both have images in $[1,2)$.
For a perfect square in the first binade to be at a midpoint, we have significands 
$A, B \in (2^{n-1}, 2^n)$
such that

\begin{equation}\label{eq:sqrat}
\frac{A}{2^{n-1}} = \left(\frac{B+1/2}{2^{n-1}} \right)^2 = \frac{(B + 1/2) \cdot (B + 1/2)}{4^{n-1}}.
\end{equation}

If there does exist a pair of integers $(A, B)$ that satisfies Equation~\eqref{eq:sqrat},
then we have found a perfect square whose square root is at a midpoint.
However, we prove the converse: that there are no roots of perfect squares on a midpoint.

\begin{lemma}
  $\forall A \in (2^{n-1}, 2^n)$, $B \in (2^{n-1}, 2^n), \frac{A}{2^{n-1}} \neq \frac{(B+1/2) \cdot (B+1/2)}{4^{n-1}}$
\end{lemma}

\begin{proof}
To prove this by contradiction, 
assume that such $A$ and $B$ exist.
Then

\begin{align*}
\frac{A}{2^{n-1}} &= \frac{B^2}{4^{n-1}} + 2\frac{B/2}{4^{n-1}} + \frac{1}{4^{n+1}}\\
\iff \frac{A}{2^{n-1}} &= \frac{B^2}{4^{n-1}} + \frac{B}{4^{n-1}} + \frac{1}{4^{n+1}}\\
\iff 2^{n+3}A &= 16B^2 + 16B + 1
\end{align*}

Which is impossible, because $2$ divides all of the terms except the $1$.
\end{proof}

For inputs in the second binade,
$[2, 4)$, we also have to prove the same lemma in the next interval.
The only difference is an extra power of two due to the different exponent on the input.

\begin{lemma}
$\forall A \in (2^{n-1}, 2^n)$, $B \in (2^{n-1}, 2^n), \frac12\cdot\frac{A}{2^{n-1}} \neq \frac{(B+1/2) \cdot (B+1/2)}{4^{n-1}}$
\end{lemma}

\begin{proof}
We can prove by contradiction on the representability of $B$:

\begin{align*}
\frac{A}{2^{n-2}} &= \frac{B \cdot B}{4^{n-1}} + 2\frac{B/2}{4^{n-1}} + \frac{1}{4^{n+1}}\\
\iff \frac{A}{2^{n-2}} &= \frac{B^2}{4^{n-1}} + \frac{B}{4^{n-1}} + \frac{1}{4^{n+1}}\\
\iff 2^{n+4}A &= 16B^2  + 16B + 1
\end{align*}

Which is impossible, because $2$ divides all of the terms except the $1$.
\end{proof}

\newpage
Together with the observation that input exponents determine output exponents on the interior of the binades and that the boundaries, being powers of two 
are either the squares of powers of two (for even powers) or irrationals (for odd powers), these lemmas prove the following theorem:

\begin{theorem}
    The square
    root of a floating-point number which is not a power of two cannot be the midpoint between
    two representable numbers.
\end{theorem}

\medskip
\section{Correction Algorithm\label{sect:implementation}}
Supported by the theorems above, and inspired by SRT division \cite{robertson1958new,tocher1958techniques},
we have devised an algorithm for correcting errors in floating-point reciprocal, division and square root estimates. 
We call this algorithm the Final Correction method 
because it produces correctly-rounded results,
so no final rounding step is needed. 
It is a residual-based correction method, and therefore applicable to all algebraic functions and other functions
for which residuals can be calculated efficiently.
In the case of reciprocal, division and square root,
only addition, subtraction, and multiplication are required, but for performance estimates, we will measure in multiples of the computation required for a fused multiply-add operation, 
since many processors are designed with the pipelinability and minimal latency of these units in mind.

We will explain our algorithm using the reciprocal function as an example.
We start with an approximation of the reciprocal of a given input $x$ in the target precision, which
we denote as $approx(\frac{1}{x})$.
Note that a reduced-precision estimate can easily be extended with zero bits, 
or a higher-precision estimate can be truncated to the target precision.

Since reciprocal maps each binade to another binade,
with the endpoints (powers of two) mapping to other endpoints,
we will assume the input is not a power of two.
We further assume that the approximation is in the correct binade.
This may be automatic for high-accuracy approximations,
or approximations for which the approximation for inputs near the endpoints of
the binade is known to fall inside the correct output binade,
or for fixed-point representations that extend to values outside the binade,
or it may require projecting approximations outside the correct binade to the first or last representable value in the binade,
but this is easy to do since the output binade is determined by the input exponent.  
As a result, we can assume that the ulp values for floating-point numbers between the approximation and correctly-rounded result are the same,
and refer to this value as $ulp$.
In the round-to-nearest-even case, the intervals of equal correction are separated by 
the values 
\[
\left\{ \frac{1}{x} - (k+\frac12)ulp\ :\ k \in \mathbb{Z} \right\},
\]
as shown in Figure~\ref{fig:numberline}.
The problem is that $1/x$ is not representable, in general, so neither are the boundary points we need to test against.
This is where the residual
\[
r = 1 - approx(\frac{1}{x}) \cdot x
\]
comes in.
The residual is representable in double the target precision
because it is calculated using one multiplication of numbers representable in the target precision.
So all the tick marks in the blue number line at the bottom of Figure~\ref{fig:numberline} are representable in double the target precision.
The mapping between number lines is linear,
so we can do our comparisons using the blue number line, against multiples of 
$\delta = x \cdot ulp$.

\begin{figure}
  \centering
  \includegraphics[width=0.5\textwidth]{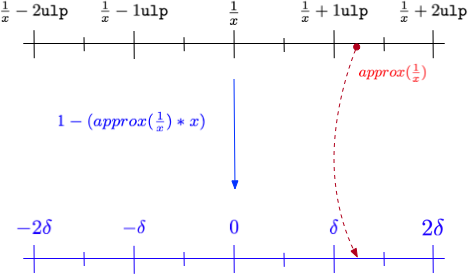}
  \caption{Number line of the ulp error and $\delta$}
  \label{fig:numberline}
\end{figure}

For practical reasons,
at this point, we must assume that we have upper and lower bounds on the error or equivalently on the residual.
For example, if we knew that all errors had magnitude at most $3ulp$,
then the residual must lie in one of the intervals 
\begin{equation} \label{eq:intervals}
(-3.5\delta, -2.5\delta), (-2.5\delta, -1.5\delta), (-1.5\delta, -0.5\delta), (-0.5\delta, 0.5\delta),
(0.5\delta, 1.5\delta), (1.5\delta, 2.5\delta), (2.5\delta, 3.5\delta).
\end{equation}
Note that these intervals
are defined with round-to-nearest modes in mind; the divisions would be
at the whole numbers and not at the midpoints for other rounding modes.
Note also that the intervals are open,
which means we never need to test for equality.
This is because the boundary values correspond to values $1/x+(k+0.5)ulp$ on the black number line,
but by Corollary~\ref{cor:recipMidpoint}, 
the values $1/x$ (for representable $x$) are never midpoints;
therefore, the values $1/x+(k+0.5)ulp$ are never representable,
and cannot be $approx(\frac{1}{x})$ which by definition is representable.

In this situation, 
we can quickly eliminate all but two adjacent intervals using the leading bits of the residual $r$,
and since this operation depends only on the bound on the error,
we could implement this as a table lookup, or as a small integer multiply-add using the leading bits of $r$ and $approx_recip$.

Putting this all together:

\begin{enumerate}
\item We have an approximation in the target precision, $approx(\frac{1}{x})$, with a (small) bound on the error.
\item We calculate the residual from the input and the approximation, and use it
  to look up a correction factor $c$, to be used in the final correction.
  We also use the leading bits of the residual to determine an adjacent pair of intervals (from \eqref{eq:intervals}) that contain the residual. 
  The boundary between the pair of intervals is called the branch point.
\item Based on a comparison between $r$ and the branch point,
we correct by either $c$ or $c+1$ ulps.
\end{enumerate}
We can encapsulate this procedure as

\begin{algorithm}[H]
  \caption{Final Correction Procedure}
  \label{alg:finalroundinginit}
  \begin{algorithmic}
  \Ensure $result = round(\frac{1}{x})$
    \State $r \gets 1 - x * approx\_recip$ 
    \State $c \gets table[ leading\_bits(r*approx\_recip) ]$ 
        \Comment{$c\in$[-5,-3,-1,1,3,5], values on blue number line}
    \State $branch\_point \gets c * 0.5 * (ulp(1/x) * x) $ \Comment{this is a
      fixed shift of x} 
    \If{$r < branch\_point$} 
      \State $result \gets approx\_recip + c$ 
    \Else
      \State $result \gets approx\_recip + c + 1$
    \EndIf
  \end{algorithmic}
\end{algorithm}

Note that without the theorems of the previous section
we would need to test for equality as well as inequality,
and in the equality case adjust $c$ based on the last bit of $approx(1/x)$ and multiple bits of $c$.

Algorithm~\ref{alg:finalroundinginit} describes the method that works for correcting
both positive and negative ulp errors. 
However, we can simplify hardware implementations by ensuring that all corrections are positive.
Since we know the maximum positive error in ulps, 
we can subtract this value from the estimate before beginning this calculation (or, even better, as part of the approximation process itself to avoid additional costs).
This would result in only non-negative values for the lookup value $c$ and the calculated $r$.

\subsection{Simplifications}

The above algorithm is applicable to other functions and different estimates,
but if we assume that this correction will be implemented in hardware,
it is likely that the estimate would also be implemented in hardware,
and so we can make adjustments to the estimate to simplify the overall size of the circuits.
If we can assume that the estimate is always an underestimate in the correct binade,
then the correction will always be a positive number of ulps and there is only one ulp size
for both the estimate and the correctly rounded result.
We have a patent pending covering one efficient way of calculating such an estimate.

With these assumptions, we arrive at Algorithm~\ref{alg:finalroundingoptim},
in which we remove the multiplication by $0.5$ of $ulp(1/x) * x$. This ensures
we will not drop any relevant bits in the calculation when later performed using fixed-point operations. Furthermore,
we do not use a table lookup to calculate $c$; instead, we apply rounding when extracting the leading bits of the product $r * approx\_recip$. 
The branch point
is obtained by multiplying $ulp\_x$ with $2 *c + 1$. Finally,
the residual is multiplied by $2$ before the comparison as $ulp_x$ is no
longer multiplied by $0.5$ in $branch\_point$.

Algorithm~\ref{alg:finalroundingoptim2} performs the identical computation of
Algorithm~\ref{alg:finalroundingoptim} except it is now done with a fixed
point representation as described in Section~\ref{sect:theorems}. For
brevity, we have also renamed the variable names, and this follows
our C implementation of the algorithm. Furthermore, we adopt the notation
$X_{(p,f)}$, where a variable $X$ has $p$ total bits and $f$ bits in the fraction 
part. The indexing notation $X[u:l]$ extracts the bits of $X$ from position
$u$ to position $l$, where positions are in little-endian order (i.e., the place-value of the lower bit is $2^l$). 
This notation is consistent with the usual visualization of bits 
as binary numbers with place values going from high to low as you go from 
left to right, and can be used in a vhdl specification of the circuits.

%

\begin{algorithm}[H]
  \caption{Simplified Final Correction for $\operatorname{round}(\frac1x)$}
  \label{alg:finalroundingoptim}
  \begin{algorithmic}
    \State $r \gets 1 - x * approx\_recip$ \Comment{By assumption $r\ge0$}
    \State $ulp\_x \gets (ulp(1/x) * x) $
      \Comment{This is a fixed shift of $x$}
    \State $c \gets leading\_bits(r*approx\_recip)$ 
    \Comment{$c\in\{0, 1, 2, 3, 4, 5\}$}
    \State $branch\_point \gets (2*c+1) * ulp\_x$ 
    \If{$2*r < branch\_point$} 
      \State $result \gets approx\_recip + c$   
    \Else
      \State $result \gets approx\_recip + c + 1$  
    \EndIf
  \end{algorithmic}
\end{algorithm}


\begin{algorithm}[H]
  \caption{Simplified Final Correction for $\operatorname{round}(\frac1x)$ with fixed point representation}
  \label{alg:finalroundingoptim2}
  \begin{algorithmic}
    \State $R_{(48,47)} \gets (1*2^{47})_{(48,47)} - X_{(24,23)} * Y_{(24,24)}$ \Comment{Fig~\ref{fig:circuit}(2)}
    \State $C_{(4,24)} \gets 
    \left(R[26:22]*
          Y[23:20]+
          2^4_{(5,29)}\right) [7:5] $ 
    \Comment{Fig~\ref{fig:circuit}(3)}
    \State $B_{(28,47)} \gets (2*C_{(4,24)} + 1) * X_{(24,23)}$ 
      \Comment{Fig~\ref{fig:circuit}(5)}
    \If{$2 * R_{(48,47)} < B_{(28,47)}$} \Comment{Fig~\ref{fig:circuit}(6)}
      \State $result \gets Y_{(24,24)} + C_{(4,24)}$   \Comment{Fig~\ref{fig:circuit}(5)}
    \Else
      \State $result \gets Y_{(24,24)} + C_{(4,24)} + 1$  \Comment{Fig~\ref{fig:circuit}(4)}
    \EndIf
  \end{algorithmic}
\end{algorithm}

\begin{figure}[h!]
  \centering
  \includegraphics[width=0.65\textwidth]{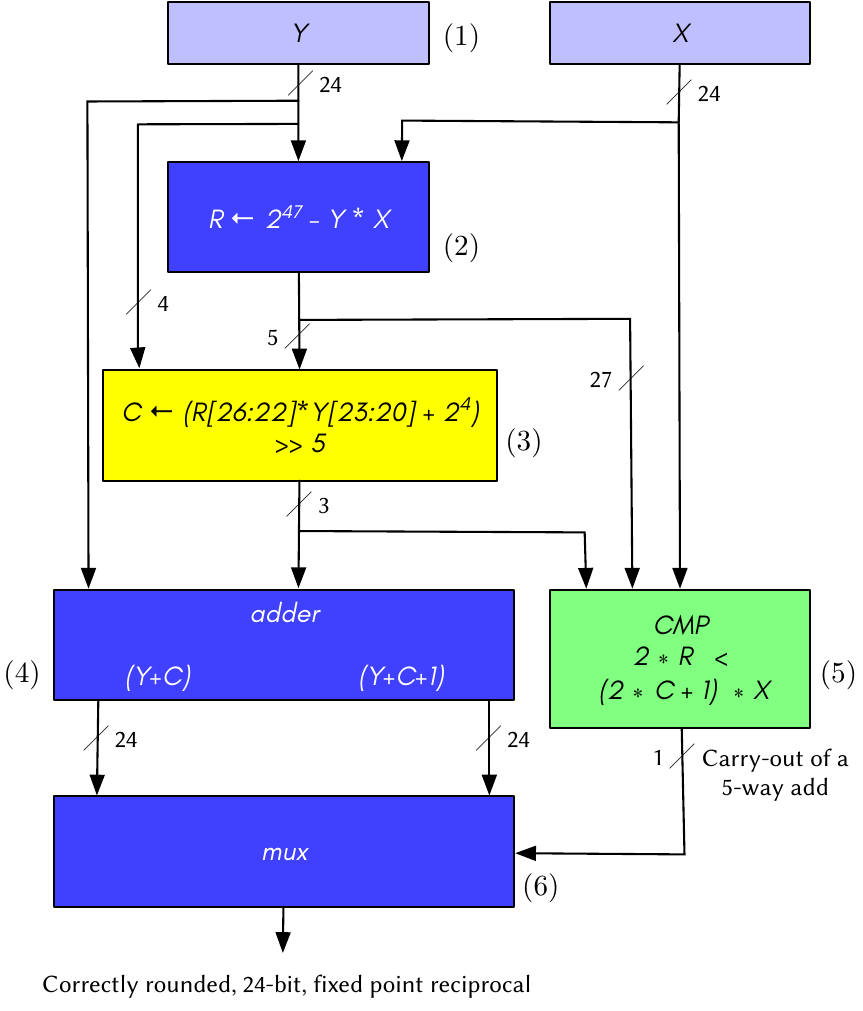}
  \caption{Circuit diagram of the Final Correction method for single-precision reciprocal}
  \label{fig:circuit}
\end{figure}

This algorithm can be translated into a very efficient hardware implementation using commonly available
circuit blocks, which we show in Figure~\ref{fig:circuit}. 
This block structure would work for any precision,
but the numbers of bits and bits used in the multiplications 
are specific to IEEE single-precision reciprocal.
It has the properties that
\begin{enumerate}
\item The reciprocal approximation comes rounded to the target precision (e.g., 23-bits).
\item The residual, $R$, is calculated with fixed-point arithmetic with twice the target precision, i.e., the natural precision for the multiplication it contains.  
  In the algorithm, it is labeled by this naively calculated precision,
  but in the C implementation below, 
  we take advantage of preconditions on $Y$ (namely that it is an approximate inverse of $X$) and assert that $R$ in fact contains at most 27 non-zero bits, due to cancellation in the subtraction.
\item The correction factor $C$ is calculated by multiplying the leading bits of the residual and the leading bits of the approximation $Y$. 
  Again, in the algorithm, we label $C$ with the maximum number of bits 
  given the operations, but preconditions imply that it has at most $3$ non-zero bits, and we check this with an assertion in the C code below.
\item Both candidate results can be computed using a special adder block, which simultaneously adds $C$ and $C + 1$
  to the approximation.
\item Since $C$ only has 3 bits,
  the comparison can be performed by a 5-way adder, adding the four partial
  products and subtracting $2*R$, generating a selector bit
  for the next block;
\item The multiplexer selects the result based on the selector bit and outputs
  a correctly rounded, fixed-point result.
\end{enumerate}

Note that the possible values of $C$, the size of the multiplication, and the possible presence of an addend for rounding in Fig~\ref{fig:circuit} (3) the size of the 4-way add used for comparison in Fig~\ref{fig:circuit} (6)
all depend on the maximum correction required, and
therefore, the accuracy of the approximation. 

\medskip
The C code below details the implementation of 
Algorithm~\ref{alg:finalroundingoptim2} using unsigned 64-bit integers.
To keep track of the structure of these fixed-point
numbers, we use the convention
\texttt{X\_p\_f} denoting a variable \texttt{X} whose 
least-significant
\texttt{p} bits may be nonzero, 
with the least-significant \texttt{f} bits being the fraction part.
If \texttt{f}$>$\texttt{p}, it means that the binary point is
to the left of the most-significant bit in the fixed-point number.
In other words, some of the leading fraction bits are structurally zero,
and don't need to be stored.
Assertions are inserted to record information about the variables.
Hardware implementers can use this information to 
optimize the widths of the data paths.
For example,
we know that the calculation of the residual \texttt{R\_48\_47}
involves significant cancellation.
To match the code with the block diagram,
comments are inserted with numbers corresponding to
individual blocks in Figure~\ref{fig:circuit}.

\begin{lstlisting}[breaklines=true,language=C,lineskip=-1pt,mathescape=true]
unsigned long long recip_final_correction(
  unsigned long long X_24_23, // binary point at position 23 
                              // interval [1, 2)
  unsigned long long Y_24_24  // interval (0.5, 1]                       // Fig 4(1)
                                          ) {
  // x == 1 <=> 1/x == 1
  if (X_24_23 == ONE_24_23) {
    return ONE_24_23;
  }

  // HALF_24_24 has shifted place value; it shares the same bit pattern as ONE_24_23
  unsigned long long HALF_24_24 = ONE_24_23;
  // approximation is left of the interval [0.5, 1)
  if (Y_24_24 < HALF_24_24) {
    Y_24_24 = HALF_24_24;
  }

  // Precondition:  Y > 1/X, ie R_48_47 is positive
  assert(ONE_48_47 >= X_24_23 * Y_24_24);
  unsigned long long R_48_47 = ONE_48_47 - X_24_23 * Y_24_24;            // Fig 4(2)
  
  // Since Y is with 8ulps of 1/X, we know
  assert(R_48_47 < (1ull << 27));
  unsigned long long R_27_47 = R_48_47;
  
  // 5bit-by-4bit mult, with rounding (+ (1 << 4)) = up to and including 5 ulps
  //        C = ((R_27_47 >> 22) * (Y_24_24 >> 20) + (1 << 4)) >> 5;
  unsigned long long R_5_25 = R_27_47 >> 22;
  unsigned long long Y_4_4   = Y_24_24 >> 20;
  assert(R_5_25 < 32);
  assert(Y_4_4 < 16);
  unsigned long long C_4_24 =
    (R_5_25 * Y_4_4 + (1 << 4)) >> 5;
  // Precondition:  Y - 1/X < 8ulps
  assert(C_4_24 < 8);
  unsigned long long C_3_24 = C_4_24;                                   // Fig 4(3)

  unsigned long long B_28_47 = ((2 * C_3_24 + 1) * X_24_23);

  // R_27_47 * 2 = R_28_47
  if ((R_27_47 * 2) < B_28_47) {                                         // Fig 4(5,6)
    return Y_24_24 + C_3_24;                                             // Fig 4(4)
  } else {
    return Y_24_24 + (C_3_24 + 1);                                       // Fig 4(4)
  }
}
\end{lstlisting}

The code above performs a 5bit-by-4bit multiplication, followed 
by rounding, to correct up to and including 7 ulps of error. 
Appendix~B lists two other versions
of the correction, one which performs 4-by-3
multiplication with no rounding and corrects errors up to
and including 3 ulps, and another which performs
5-by-3 multiplication with rounding, and corrects errors up
to and including 6 ulps.

\subsection{Testing}

For single precision (and lower) it is easy to exhaustively test all possible inputs,
including all possible errors.
We have tested our C implementation exhaustively on all single-precision values in the open interval $(1.0,2.0)$ 
with all errors in the set $\{-7ulp,-6ulp,-5ulp,-4ulp,-3ulp,-2ulp,-1ulp,0ulp\}$
for which the resulting input to be tested is in the same interval (i.e., binade),
and verified that the correctly rounded result is returned in all cases.
All code for testing is included in Appendix A, 
and can be compiled with the code in the previous section.

\section{Prior Work\label{sect:priorwork}}
As already stated in Section~\ref{sect:theorems}, proofs of the properties of
midpoint and exact results have been discovered independently by several others.
We specialize our proofs to the functions used in our implementations. In his
work, \citet{marksteinComputationElementaryFunctions1990a} also uses
the fact that midpoints do not exist for division/square root to deduce that
two \texttt{fma} instructions, one for computing the residual, and one for
applying the correction can yield a correctly-rounded result assuming the error
is at most one ulp. Our work extends this and provides a more powerful and more efficient correction procedure that extends the tolerance of the error bounds, which can be
adjusted according to the accuracy of the approximation step,
at the expense of introducing a new instruction.

Correctly rounding elementary functions is inextricably linked to the precision
of the computation. The theorems
presented in \citet{brisebarreCorrectRoundingAlgebraic2007b} and
\citet{iordacheInfinitelyPreciseRounding1999}
make clear the bounds on accuracy for which to perform intermediate calculations
in order to correctly round these functions. Though they focus on providing
these bounds in great detail, they do not make explicit that any approximation
does not fall on the exact intervals. Their work is complementary to our approach,
as the results can be used in the step used to generate the approximation.

\citet{harrisonMachineCheckedTheoryFloating1999b} developed a landmark in automatic
proof checking by using HOL-Light as a proof assistant for
verifying theorems about floating-point operations. This work created a new level of confidence for such computation, and was later extended by Harrison to verify implementations of division
\cite{harrisonFormalVerificationIA642000a} and trigonometric functions
\cite{harrisonFormalVerificationFloating2000}. This work was motivated by the
very costly Intel Pentium FDIV bug \cite{pricePentiumFDIVFlawlessons1995}, which
drove the demand for verification of floating-point algorithms. While our work
does not include formal verification, the properties identified and arguments made lend themselves to formalization, and could be incorporated into a formal proof of the correctness of a floating point unit.

\section{Conclusions and future work}

In this paper, we proved theorems about the irrepresentability of results for floating point reciprocal, division and square root functions, 
and immediately used them to simplify the computation of floating-point corrections
to estimates of these functions.
These results are applicable to all rounding modes, 
and can be extended to other algebraic functions (functions definable as the solution of an algebraic equation).

Some of the irrepresentability results were known, but the more general result that reciprocals are not representable even in higher (and even slightly lower) precision, and the tightness of the bound, is new.
Such results would be useful in implementing operations with mixed precision, e.g., dividing a single-precision value by a double-precision value and returning a single-precision value. 

In future work, we plan to extend our proofs to other non-IEEE-standard functions,
such as reciprocal square root, and consequently extend
the Final Correction algorithm for these cases, including an optimized
hardware implementation of these functions.
The results should also be extended to mixed-precision cases of interest in 
machine learning and mathematical optimization.

While the focus of this paper is to present a novel correction scheme to replace final rounding, and to present the enabling irrepresentability theorems, we have not compared the performance of our method to approximation and rounding approaches in the literature. 
We anticipate that the highest impact applications of this approach will be in floating-point units of general-purpose processors and specialized co-processors.
Since the most relevant processor families do not make architectural details public,
performance comparisons for such processors is not possible.
It would, however, be feasible to compare performance when implemented
in an FPGA.
For example,
\citet{goldbergFPGAImplementationPipelined2007},
 report on an FPGA implementation of IEEE double-precision
division. In particular, it reports on the latencies and hardware resources
used for the implementation.

TODO ******  check yellow box / code / algorithm 3

\section*{Acknowledgement}
We acknowledge the support of the Natural Sciences and Engineering Research Council of Canada (NSERC), funding reference number CRDPJ-536628-2018.




\bibliographystyle{ACM-Reference-Format}
%

\bibliography{finalCorrection}

\section*{Appendix A:  Testing Harness}

\begin{lstlisting}[breaklines=true,language=C,lineskip=-1pt,mathescape=true]
#include <stdio.h>
#include <stdlib.h>
#include <math.h>
#include <stdbool.h>
#include <string.h>
#include <assert.h>

// Forward declarations
unsigned long long recip_final_correction(unsigned long long x,
                                          unsigned long long approxRecip);
void batch_mode();
unsigned long long individual_mode(unsigned long long x, unsigned long long ulp_err);

float fixed23_to_float_given_exp(unsigned long long int x, int unbiased_exp);
unsigned long long float_to_fixed23_given_exp(float x, int unbiased_exp);

// This is the IEEE single precision including the implied one bit
const unsigned int n = 24;

// n - 1 = 23 fraction bits
const unsigned long long ONE_24_23 = (1ull << (n - 1));
const unsigned long long ONE_48_47 = 1ull << (n + (n - 1));
const unsigned long long TWO = (ONE_24_23 << 1);

int main(int argc, const char *argv[]) {
  // Get mode to run: Batch or individual mode
  if (argc == 2 && strcmp(argv[1], "-b") == 0) {
    batch_mode();
  }
  else if (argc == 4 && strcmp(argv[1], "-i") == 0) {
    unsigned long long x;
    unsigned long long ulp_err;
    sscanf(argv[2], "%llu", &x);
    sscanf(argv[3], "%llu", &ulp_err);
    unsigned long long output = individual_mode(x, ulp_err);
    printf("%lf\n", (double)fixed23_to_float_given_exp(output, -1));
  } else {
    unsigned long long x = 0xaaaaaa;
    // 0xffffff; // float_to_fixed23_given_exp(1.999, 0);
    // printf("x = 0x%016llx\n", x);
    for (int ulp_err = -1; ulp_err >= -7; ulp_err--) {
      printf("\nulp_err = %d\n", ulp_err);
      individual_mode(x, ulp_err);
    }
  }
}

void batch_mode(int ulp_err) {
  // to test other versions, change the bound
  for (int ulp_err = 0; ulp_err <= 7; ulp_err++) {
    printf("Testing for ulp_err = %d\n", ulp_err);
    for (unsigned long long X_24_23 = ONE_24_23 + 1; X_24_23 < TWO; X_24_23++) {
      float x_float = fixed23_to_float_given_exp(X_24_23, 0);
      assert(1.0 <= x_float && x_float < 2.0);

      unsigned long long true_recip =
        float_to_fixed23_given_exp(1.0f/x_float, -1);
      assert(0.5f <= 1.0f/x_float && 1.0f/x_float < 1.0);

      unsigned long long approx_recip = true_recip - ulp_err;
      unsigned long long corrected_recip;
      unsigned long long residual =
        ONE_48_47 - X_24_23 * approx_recip;

      corrected_recip =
        // to test other versions, change this call
        recip_final_correction(X_24_23, approx_recip);

      if (true_recip != corrected_recip) {
        printf("x_float = %f\nx_24bits_Em23 = 0x%016llx\ntrue_recip = 0x%016llx\ncorrected_recip = 0x%016llx\n",
               x_float, X_24_23, true_recip, corrected_recip);
      }

      // assert(true_recip == corrected_recip);
    }
  }
  // printf("correct_results = %d\nincorrect_results = %d\n",
  //      correct_results, incorrect_results);
}

unsigned long long individual_mode(unsigned long long x, unsigned long long ulp_err) {
  // Input precondition: 1 <= x < 2, in fixed point
  assert(ONE_24_23 <= x && x < TWO);

  // Ensure precondition is preserved when converted to float
  float x_float = fixed23_to_float_given_exp(x, 0);
  assert(1.0 <= x_float && x_float < 2.0);

  // Ensure 0.5 <= x < 1 of true_recip when converted back to fixed point
  unsigned long long true_recip = float_to_fixed23_given_exp(1.0f/x_float, -1);
  assert(0.5 <= fixed23_to_float_given_exp(true_recip, -1)
         && fixed23_to_float_given_exp(true_recip, -1) < 1.0);

  unsigned long long corrected_val = recip_final_correction(x, true_recip + ulp_err);

  // Ensure correction is equal to true value
  if (corrected_val != true_recip) {
    printf("Correction failed: \n");
    printf("  corrected_val:      0x%016llx\n", corrected_val);
    printf("  true_recip:         0x%016llx\n", true_recip);
  }
  assert(corrected_val == true_recip);

  return corrected_val;
}

unsigned long long float_to_fixed23_given_exp(float x, int unbiased_exp) {
  typedef union {float f; int i;} float_union_t;
  float_union_t ux;
  ux.f = x;
  int exp = (ux.i & 0x7F800000) >> 23;
  exp -= 127;
  assert(exp == unbiased_exp);
  unsigned long long significand = (1ull << 23) /*implybit*/ ;
  significand += ux.i & 0x007FFFFF;
  return significand;
}

float fixed23_to_float_given_exp (unsigned long long x, int unbiased_exp) {
  float result;
  result = ((float)x)/pow(2.0,23 - unbiased_exp);
  return result;
}

\end{lstlisting}

\section*{Appendix B: Listings for Different Multiplication Widths\label{sect:appendixB}}

4-by-3 multiplication, no rounding, correction up to and including 3 ulps of error:
\begin{lstlisting}[breaklines=true,language=C,lineskip=-1pt,mathescape=true]
unsigned long long recip_final_correction(
  unsigned long long X_24_23, // binary point at position 23 // interval [1, 2)
  unsigned long long Y_24_24 // interval (0.5, 1]
                                          ) {
  // x == 1 <=> 1/x == 1
  if (X_24_23 == ONE_24_23) {
    return ONE_24_23;
  }

  // HALF_24_24 has shifted place value; it shares the same bit pattern as ONE_24_23
  unsigned long long HALF_24_24 = ONE_24_23;
  // approximation is left of the interval [0.5, 1)
  if (Y_24_24 < HALF_24_24) {
    Y_24_24 = HALF_24_24;
  }

  // { R_48_47 : Int | R_48_47 is positive }
  assert(ONE_48_47 >= X_24_23 * Y_24_24);
  unsigned long long R_48_47 = ONE_48_47 - X_24_23 * Y_24_24;
  assert(R_48_47 < (1ull << 27) /* 8ulps */);
  unsigned long long R_26_47 = R_48_47;

  // We know R_48_47/X_24_23 < 8ulps
  // So we can refine R_48_47
  // 

  // Different multiplications can correct different levels of errors, e.g.,
  // 3bit-by-3bit mult, no rounding = up to and including 3 ulps
  //        C = ((R_48_47 >> 21) * (Y_24_24 >> 21)) >> 5;

  unsigned long long R_27_26 = R_48_47 >> 21;
  unsigned long long Y_3_3   = Y_24_24 >> 21;
  assert(R_27_26 < 16);
  assert(Y_3_3 < 8);
  unsigned long long C_25_24 =
    (R_27_26 * Y_3_3) >> 5;
  assert(C_25_24 <= 3);
  unsigned long long C_2_24 = C_25_24;

  unsigned long long B_27_47 = ((2 * C_2_24 + 1) * X_24_23);

  // R_26_47 * 2 = R_27_47
  if ((R_26_47 * 2) < B_27_47) {
    return Y_24_24 + C_2_24;
  } else {
    return Y_24_24 + (C_2_24 + 1);
  }
}
\end{lstlisting}

\noindent
5-by-3 multiplication, with rounding $(+ (1 << 4))$, correction up to and including 6 ulps of error:
\begin{lstlisting}[breaklines=true,language=C,lineskip=-1pt,mathescape=true]
unsigned long long recip_final_correction(
  unsigned long long X_24_23, // binary point at position 23 // interval [1, 2)
  unsigned long long Y_24_24 // interval (0.5, 1]
                                          ) {
  // x == 1 <=> 1/x == 1
  if (X_24_23 == ONE_24_23) {
    return ONE_24_23;
  }

  // HALF_24_24 has shifted place value; it shares the same bit pattern as ONE_24_23
  unsigned long long HALF_24_24 = ONE_24_23;
  // approximation is left of the interval [0.5, 1)
  if (Y_24_24 < HALF_24_24) {
    Y_24_24 = HALF_24_24;
  }

  // { R_48_47 : Int | R_48_47 is positive }
  assert(ONE_48_47 >= X_24_23 * Y_24_24);
  unsigned long long R_48_47 = ONE_48_47 - X_24_23 * Y_24_24;
  assert(R_48_47 < (1ull << 27) /* 8ulps */);
  unsigned long long R_27_47 = R_48_47;

  // We know R_48_47/X_24_23 < 8ulps
  // So we can refine R_48_47
  // 

  // Different multiplications can correct different levels of errors, e.g.,
  // 3bit-by-3bit mult, no rounding = up to and including 3 ulps
  //        C = ((R_48_47 >> 21) * (Y_24_24 >> 21)) >> 5;
  // 3bit-by-3bit mult, with rounding (+ (1 << 4)) = up to and including 5 ulps
  //        C = ((R_48_47 >> 21) * (Y_24_24 >> 21) + (1 << 4)) >> 5;
  // 4bit-by-4bit mult, with rounding (+ (1 << 6))
  //   = up to and including 7 ulps
  // 4-by-3 mult, no rounding, correct up to 3ulps (Only change is from 3 to 4 bit for the left operand). Can assert R_48_47 less than 4 ulps instead of 8 ulps. Can also assert C_25_24 < 4, giving C_2_24
  // 5-by-3 mult, with rounding, correct up to 6ulps (Left operand from 3 to 5 bits, ulp correct from 5 to 6)

  unsigned long long R_27_26 = R_48_47 >> 21;
  unsigned long long Y_3_3   = Y_24_24 >> 21;
  assert(R_27_26 < 32);
  assert(Y_3_3 < 8);
  unsigned long long C_25_24 =
    (R_27_26 * Y_3_3 + (1 << 4)) >> 5;
  assert(C_25_24 < 8);
  unsigned long long C_3_24 = C_25_24;


  unsigned long long B_28_47 = ((2 * C_3_24 + 1) * X_24_23);

  // R_27_47 * 2 = R_28_47
  if ((R_27_47 * 2) < B_28_47) {
    return Y_24_24 + C_3_24;
  } else {
    return Y_24_24 + (C_3_24 + 1);
  }
}
\end{lstlisting}

\end{document}